\documentclass[conference]{IEEEtran}
\IEEEoverridecommandlockouts
\usepackage{footnote}
\usepackage[utf8]{inputenc}
\usepackage[english]{babel}
\usepackage[T1]{fontenc}
\usepackage{amsmath, amssymb, amsthm}
\usepackage{mathtools}
\usepackage{braket}
\usepackage{halloweenmath}
\usepackage{cite}
\usepackage{graphicx}
\usepackage{xcolor}
\usepackage[dvipsnames]{xcolor}
\usepackage{tcolorbox}
\usepackage{tabularx}
\usepackage{colortbl}
\usepackage{tikz}
\usepackage{enumitem}
\usepackage{textcomp}
\usepackage{algorithm}
\usepackage{algpseudocode}
\usepackage{float}
\usepackage{placeins}
\usepackage{subfigure}
\usepackage{caption}
\usepackage{subcaption}
\usepackage[left=0.64in,right=0.64in,top=0.72in,bottom=1.07in]{geometry}
\newtheorem{theorem}{Theorem}
\newtheorem{lemma}{Lemma}
\newtheorem{corollary}{Corollary}

\newtheorem{example}{Example}
\newtheorem{remark}{Remark}

\newlength{\myeqskip} 
\AtBeginDocument{%
    \setlength\abovedisplayskip{\myeqskip}%
    \setlength\belowdisplayskip{\myeqskip}%
    \setlength\abovedisplayshortskip{\myeqskip-\baselineskip}%
    \setlength\belowdisplayshortskip{\myeqskip}}
\def\BibTeX{{\rm B\kern-.05em{\sc i\kern-.025em b}\kern-.08em
		T\kern-.1667em\lower.7ex\hbox{E}\kern-.125emX}}

\allowdisplaybreaks
\begin{document}


\title{Non-Binary Quasi-Cyclic LDPC Codes with Entanglement Assistance}
 \author{\IEEEauthorblockN{Pavan Kumar and Shayan  Srinivasa Garani}
  \IEEEauthorblockA{Department of Electronic Systems Engineering, Indian Institute of Science, Bengaluru-560012, India\\
    Emails:\{pavankumar1, shayangs\}@iisc.ac.in}
 }
\maketitle
\begin{abstract}
We construct two families of non-binary entanglement assisted (EA) quasi-cyclic (QC) quantum low-density parity-check (QLDPC) codes over arbitrary finite fields, each possessing a precisely determined code rate.
 The first family is derived from a pair of non-binary classical QC-LDPC codes, designed such that the unassisted portion of the overall Tanner graph of the resulting EA-QC-QLDPC code is free of 4-cycles. The second family, on the other hand, is constructed from a single non-binary classical QC-LDPC code whose Tanner graph itself is 4-cycle-free. In developing the codes belonging to the first family, we employ a \emph{single Bell pair} to establish entanglement between the transmitter and the receiver, thereby minimizing the required entanglement resources. Furthermore, these constructions demonstrate that careful graph-based design can effectively balance error-correction performance with entanglement consumption, providing a practical approach for realizing efficient non-binary EA-QC-QLDPC codes.

 \end{abstract}
\section{Introduction}
LDPC codes were first introduced by Gallager in 1963 \cite{gallager1963low}. There was renewed interest in these codes with the advent of iterative decoding techniques. In 1998, Davey and MacKay~\cite{davey1998low} highlighted the advantages of non-binary LDPC codes, demonstrating through computer simulations that these codes could outperform their binary counterparts. This led to further studies, including the work of \cite{huang2010large}, which demonstrated that non-binary LDPC codes not only achieve superior performance for short lengths but also integrate seamlessly with high-order modulation schemes in multiple-input, multiple-output (MIMO) systems. By avoiding bit-to-symbol mapping and demapping, these codes can further enhance the code performance. There is a renewed interest in designing high-performance non-binary LDPC codes, including various aspects of decoding~\cite{costantini2012non, kang2010quasi, song2009unified, zhou2007high, chen2010two,voicila2010low,savin2008min,nozaki2012analysis, nozaki2011analysis, liu2012computing, zhao2012class}. 


The success of classical non-binary LDPC codes has naturally motivated the exploration of their quantum analogues, resulting in the development of both   unassisted non-binary quantum LDPC (QLDPC) and entanglement-assisted (EA) codes, along with decoding strategies \cite{xie2016reliable, shao2018entanglement, andriyanova2012quantum, hwang2013decoding}. Non-binary QLDPC codes provide improved error-correction performance and greater design flexibility compared to their binary counterparts. However, challenges persist in decoding complexity and in minimizing the use of entangled resources.

A key challenge in designing effective QLDPC codes relies on satisfying the dual-containment property \cite{calderbank1998quantum,ketkar2006nonbinary}, which  introduces 4-cycles in the Tanner graph, leading to message correlations that lead to slow decoding convergence and high error floors. This limitation can be overcomed using EA-QLDPC codes, where preshared entangled bits  (ebits) between the sender and receiver relax structural constraints, enabling the construction of codes with enhanced properties \cite{brun2006correcting, galindo2019entanglement}. The \emph{ebits} preshared with the receiver are assumed to be error-free and do not participate in the decoding process.

In \cite{shao2018entanglement}, the authors constructed EA-QLDPC codes over the finite field $\mathbb{F}_{2^{\ell}}$ using a single classical code. However, this approach results in the presence of 4-cycles in the unassisted portion of the overall Tanner graph. Similarly, the constructions in \cite{andriyanova2012quantum, Decoder} also lead to  QLDPC codes over $\mathbb{F}_{2^{\ell}}$ that contain 4-cycles. In contrast, \cite{xie2016reliable} presented non-CSS codes over $\mathbb{F}_{4}$ derived from a pair of classical QC-LDPC codes, avoiding 4-cycles. Nevertheless, existing works do not provide a comprehensive solution, namely, constructing QLDPC codes over arbitrary finite fields with Tanner graphs that are devoid of short cycles. This gap motivates the design of non-binary EA-QLDPC codes whose unassisted portion of overall Tanner graphs are free of $4$-cycles, while using a \emph{single ebit}.
 
Our Contributions in this paper are as follows: First we determine the exact rank of the parity-check matrix of a specific  array-based non-binary classical QC-LDPC code. Next, we propose two families of non-binary EA-QC-QLDPC codes defined over arbitrary finite fields, each with an analytically determined \textit{code rate} and required \textit{ebits}. The first family is constructed from a pair of non-binary classical QC-LDPC codes, ensuring that the unassisted portion of the overall Tanner graph of the resulting EA-QC-QLDPC code is free of 4-cycles. For a clear understanding of the term \textit{unassisted portion} of the overall Tanner graph, refer to Figure 2 in \cite{kumar2024entanglement}. The second family is derived from a single non-binary classical QC-LDPC code whose Tanner graph itself is 4-cycle-free. In designing the first family, we employ a \emph{single Bell pair} to establish entanglement between the transmitter and receiver, thereby minimizing the required entanglement resources, useful in quantum communications.


The paper is organized as follows: In Section~\ref{Sec.2}, we present the construction of non-binary QC-LDPC codes by tiling permutation matrices of prime order while ensuring an exact code rate. In Section~\ref{Sec.3}, this framework was extended to construct two families of non-binary EA-QC-QLDPC codes, one of which requires only a \emph{single ebit} to be shared between the transmitter and receiver. Finally, Section~\ref{Sec.5} provides concluding remarks summarizing the key findings of this work.
\section{Non-Binary QC-LDPC codes}\label{Sec.2}
Let \(\mathbb{F}_{q}\) be a finite field, where \(q = p^{m}\) for some prime \(p\) and positive integer \(m\). This notation for the field will be used consistently throughout the paper. For an integer $n\geq3$,  define a diagonal matrix $D^{(0)}=\text{diag}\{a_{0},a_{1},\ldots,a_{n-1}\}, \text{ where } a_{i}\in\mathbb{F}_{q}\setminus\{0\}.$
For each $1 \leq i \leq n - 1$, we define $D^{(i)}$ by cyclically shifting each row of $D^{(0)}$ $i$ positions to the right. For the matrix $D^{(i)} - D^{(j)}$, where $1 \leq i, j \leq n - 1$, it is easy to see that
\begin{equation}\label{sumofrowzero}
  \sum_{k=0}^{n-1} a_k^{-1} R_k = \mathbf{0},
\end{equation}
where \(\mathbf{0}\) denotes the zero vector over \(\mathbb{F}_{q}\) of length $n$, and $R_k$ represents the \(k^{\text{th}}\) row of the matrix \(D^{(i)} - D^{(j)}\), and the equation \eqref{sumofrowzero} will be used later in the paper.
 Next, for all $0\leq i,j\leq n-1$, we define \begin{equation}\label{matrix_product}
      D^{(i)}\circ D^{(j)}=D^{((i+j) \text{ (mod $n$}))} \text{ and }
    (D^{(i)})^{\Gamma}=D^{(n-i)}.
 \end{equation}
Throughout this paper, we adopt a slight abuse of notation to facilitate the expansion of $\prod\limits_{i=0}^{n} (D^{(a_{i})}\circ D^{(b_{i})})$ as\\
\begin{equation*}
    \prod\limits_{i=0}^{n} (D^{(a_{i})}\circ D^{(b_{i})})=(D^{(a_{0})}\circ D^{(b_{0})})\circ \cdots\circ (D^{(a_{n})}\circ D^{(b_{n})}).
\end{equation*}
For a positive integer $n$, let $\mathbb{Z}_{n}=\{0,1,\ldots,n-1\}$ be an integer ring. Let $m, t$ be two positive integers, such that $m<t$. We define a parity-check matrix $H$ and the corresponding model matrix $M$ by the following equations:
\begin{equation}\label{cycleparity}
H=\begin{bmatrix}
D^{(a_{0,0})} & D^{(a_{0,1})}& \cdots & D^{(a_{0,t})} \\
D^{(a_{1,0})} & D^{(a_{1,1})} & \cdots &  D^{(a_{1,t})} \\
\vdots & \vdots & \ddots & \vdots\\
D^{(a_{m,0})} & D^{(a_{m,1})} & \cdots & D^{(a_{m,t})} 
\end{bmatrix}, 
\end{equation}
\begin{equation}\label{modelmatrix}
 	M=\begin{bmatrix}
 		a_{0,0}&a_{0,1}&\cdots &a_{0,t}\\		a_{1,0}&a_{1,1}&\cdots&a_{1,t}\\	
        \vdots&\vdots&\ddots&\vdots\\
 		a_{m,0}&a_{m,1}&\cdots&a_{m,t}	
 	\end{bmatrix},
 \end{equation}
where $a_{i,j}\in\mathbb{Z}_{n}$, for all $i,j$. In the parity-check matrix $H$, defined in~\eqref{cycleparity}, the expressions $[D^{(a_{i,0})},\ D^{(a_{i,1})},\ \ldots,\ D^{(a_{i,t})}]$ and $[(D^{(a_{0,j})})^{T},\ (D^{(a_{1,j})})^{T},\ \ldots,\ (D^{(a_{m,j})})^{T}]^{T}
$ are referred to as $i^{\text{th}}$ \emph{block-row} and $j^{\text{th}}$ \emph{block-column} of $H$, respectively. We adopt this terminology throughout the paper.
      
 A closed path of length $2k$ in any parity-check matrix of the form in \eqref{cycleparity} is a sequence of block-row and block-column index pairs $(i_0, j_0), (i_0, j_1); (i_1, j_1), (i_1, j_2); \ldots ;(i_{k-1}, j_{k-1})$, $ (i_{k-1}, j_0)$, with $i_\ell \neq i_{\ell+1}, j_\ell \neq j_{\ell+1}$, for $\ell = 0, 1, \ldots, k-2$, and  $i_{k-1} \neq i_0, j_{k-1} \neq j_0 $.
 
 
In the following theorem, a simple generalization of Theorem 2.1 in \cite{fossorier2004quasicyclic}, we establish a general condition for the existence of 
$2k$-cycles in non-binary QC-LDPC codes with parity-check matrix $H$ in \eqref{cycleparity}.
\begin{theorem}\label{2k-cycletheorem}
Let $H$, defined in \eqref{cycleparity}, be a parity-check matrix of a QC-LDPC code $\mathcal{C}$. Then, there exists a $2k$-cycle in the Tanner graph of $\mathcal{C}$ if and only if there exists a closed path $(i_0, j_0), (i_0, j_1); (i_1, j_1), (i_1, j_2); \ldots ;(i_{k-1}, j_{k-1})$, $ (i_{k-1}, j_0)$ in $H$ such that
\begin{equation}\label{cyclcondition.1}
    \prod_{t=0}^{k-1} \left(D^{(a_{i_{t},j_{t}})}\circ (D^{(a_{i_{t},j_{t+1}})})^{\Gamma}\right)=D^{(0)}, \text{ where } j_{k}=j_{0}.
\end{equation}
\end{theorem}
\begin{proof}
    One can easily prove the lemma employing the same arguments as those given in \cite[Theorem 1]{kumar2026entanglement}.
\end{proof}

\begin{corollary}\label{4-cyclecoro}
   There does not exist a 4-cycle in the Tanner graph of a parity-check matrix $H$, defined in \eqref{cycleparity}, if and only if   \begin{equation*}
       (a_{i_{1},j_{0}}-a_{i_{0},j_{0}}) \neq (a_{i_{1},j_{1}}-a_{i_{0},j_{1}})\pmod{n},
   \end{equation*}
for all $0\leq i_{0}<i_{1}\leq m$ and $0\leq j_{0}<j_{1}\leq t$.
\end{corollary}
We now proceed to the construction of non-binary QC-LDPC codes with exact code rates.

Let $\mathbb{F}_{r} = \{0, 1, \ldots, r-1\}$ denote the finite field of order $r$, where $r$ is an odd prime such that the characteristic $p$ of the field $\mathbb{F}_{q}$, defined earlier, does not divide $r$, i.e., $p\nmid r$.

We define a matrix
\begin{equation}\label{model_matrix}
   M =
\begin{bmatrix}
b_{0,0} & b_{0,1}& \cdots & b_{0,r-1} \\
b_{1,0} & b_{1,1} & \cdots & b_{1,r-1} \\
\vdots  & \vdots    & \ddots  & \vdots \\
b_{r-1,0} & b_{r-1,1} & \cdots & b_{r-1,r-1}
\end{bmatrix},
\end{equation}
 where $b_{i,j} \in \mathbb{F}_{r}$, for all $0 \leq i,j \leq r-1$, and the entries in $M$ are populated as follows:
Assign zero along the entire row corresponding to $i=0$. Populate the row corresponding to ~$i=1$ (excluding the $1^{\mathrm{st}}$ column) with different nonzero elements randomly chosen from $\mathbb{F}_r$. Next, select a scalar $k_1 \in \mathbb{F}_r \setminus \{0, 1\}$ and multiply the row corresponding to $i=1$ by $k_{1}$ to populate the row corresponding to $i=2$. For the $i^{\text{th}}$ row, where $3\leq i\leq r-1$, choose a new scalar $k_{i-1} \in \mathbb{F}_r$ not previously used and not equal to 0 or 1, and compute the $i^{\mathrm{th}}$ row by multiplying the row associated to $i=1$ by $k_{i-1}$. 
Let $C_{M}$ denote the collection of all possible distinct $r \times r$ matrices $M$. For each $M \in C_{M}$, we define a block matrix~$\mathbf{A}_{M}$
\begin{equation}\label{A}
   \mathbf{A}_{M}=
\begin{bmatrix}
D^{(b_{0,0})} & D^{(b_{0,1})} & \cdots & D^{(b_{0,r-1})} \\
D^{(b_{1,0})} & D^{(b_{1,1})} & \cdots & D^{(b_{1,r-1})} \\
\vdots        & \vdots        & \ddots & \vdots          \\
D^{(b_{r-1,0})} & D^{(b_{r-1,1})} & \cdots & D^{(b_{r-1,r-1})}
\end{bmatrix},
\end{equation}
where $b_{i,j}$ is the $(i,j)$-th entry in $M$, and 
$M$ in \eqref{model_matrix} is referred to as the \emph{model matrix} of $\mathbf{A}_{M}$. 

Further, we define a class of block matrices associated with the class $C_{M}$ as follows:
\begin{equation}\label{classCA}
    \mathcal{C}_{A} = \left\{ \mathbf{A}_{M} : \text{ for all } M \in C_{M} \right\}.
\end{equation}

Let $M$ be a matrix in the class $C_{M}$, and let $R_i$ and $R_j$ be two distinct arbitrary rows of the matrix $M$. Then there exist distinct non-zero elements $k_i, k_j$ in $\mathbb{F}_{r}$ ($k_i\neq k_j$) such that $R_i = (0, k_i x_1, \ldots, k_i x_{r-1}) \quad \text{and} \quad R_j = (0, k_j x_1, \ldots, k_j x_{r-1}),$ where $x_{\ell} \in \mathbb{F}_{r}^{*}=\mathbb{F}_{r}\setminus\{0\}$, for all $1 \leq \ell \leq r-1$, and $x_{\lambda} \neq x_{\mu}$ for $\lambda \neq \mu$. It is easy to verify that the vector $R_i - R_j$ has all distinct entries under $\pmod{r}$. From Corollary~\ref{4-cyclecoro}, the Tanner graph of the non-binary classical QC-LDPC code with parity-check matrix $\mathbf{A}_{M}$ girth $> 4$.

To find code rate of the non-binary QC-LDPC code with a parity-check matrix {\small$\mathbf{A}_{M} \in \mathcal{C}_{A}$}, we first compute the $\mathrm{gfrank}$ over $\mathbb{F}_{q}$ of a specific block matrix ~{\small$\mathbf{H} =
   [h_{i,j}]_{0 \leq i,j \leq r-1},$} where {\small$h_{i,j} = D^{(\mathrm{mod}(ij, r))}$}, in the class $\mathcal{C}_{A}$ defined by \eqref{classCA}, and
 then discuss the {\small$\mathrm{gfrank}_{q}(\mathbf{A}_{M})$}. We need the following lemma to compute $\mathrm{gfrank}_{q}(\mathbf{H})$.
\begin{lemma}\label{lemma1}For $0\leq i,k\leq r-1$, let $R_{i}^{(k)}$ represents $i^{\text{th}}$ row in the $k^{\text{th}}$ block-row of $\mathbf{H}$, and let $\langle(v_{1},\ldots,v_{r^{2}}),(u_{1},\ldots,u_{r^{2}})\rangle=v_{1}u_{1}+\cdots+v_{r^{2}}u_{r^{2}}$, where $v_{i}$ and $u_{j}$ are the scalars from the field $\mathbb{F}_{q}$, for all $i$ and $j$. Then, we have
	\begin{center}
		$\langle R_{i}^{(k)},R_{j}^{(l)}\rangle=\begin{cases}
			r(a_{i})^{2},&\text{if } i=j\text{ and }k=l,\\
			0,&\text{if } i\neq j\text{ and }k=l,\\
			a_{i}a_{j},&\text{otherwise}.
		\end{cases}$
	\end{center}
\end{lemma}
\begin{proof} The lemma can be proved using the approach outlined in Lemma 1 of \cite{kumar2024entanglement}.
\end{proof}
\begin{remark}\label{re1}
	The above lemma also applies to any block matrix belonging to the class $\mathcal{C}_{A}$.
\end{remark}
\begin{theorem}\label{theorem1}
	Let $H$ be a submatrix of $\mathbf{H}$ consisting distinct $k$ $(1\leq k\leq r)$ block-rows. Then, $\mathrm{gfrank}(H)=r+(k-1)(r-1)$ over $\mathbb{F}_{q}$.
\end{theorem}
\begin{proof} Suppose  $k = 1.$ In this case, $\mathrm{gfrank}_{q}(H) = r$, as the matrix is already in row echelon form.

Next, consider the case when \( k = 2 \). Subtracting the first block-row from the second and applying \eqref{sumofrowzero} results in the last row of the second block-row becoming zero, i.e., any two block-rows  of $H$ are linearly dependent.

Now, consider the $k$ block-rows of $H$, and remove the last row from each block-row except the first. Let
\begin{equation} \label{linearcombi}
			\sum_{i_{1}=0}^{r-1}\alpha_{i_{1}}^{(0)}R^{(0)}_{i_{1}}+\sum_{i_{2}=0}^{r-2}\alpha_{i_{2}}^{(1)}R^{(1)}_{i_{2}}+\cdots+\sum_{i_{k}=0}^{r-2}\alpha_{i_{k}}^{(k-1)}R^{(k-1)}_{i_{k}}=0,
		\end{equation}
        where coefficients are from the field $\mathbb{F}_{q}$. 
        
        Define the inner product $\langle,\rangle:\mathbb{F}_{q}^{r^{2}}\times \mathbb{F}_{q}^{r^{2}}\to \mathbb{F}_{q}$, as follows:		\begin{equation}\label{euclideanF}		\langle(v_{1},\ldots,v_{r^{2}}),(u_{1},\ldots,u_{r^{2}})\rangle=\sum_{i=1}^{r^{2}}v_{i}u_{i}.
	\end{equation}
Taking the inner product, as defined by \eqref{euclideanF}, of \eqref{linearcombi} with each of the remaining rows from all the block-rows  $R^{(0)}, R^{(1)}, \dots, R^{(k-1)}$, and applying Lemma 1, results in the following system of linear equations:	
	\begin{align*}
	& r\alpha_{j_{1}}^{(0)}(a_{j_{1}})^{2}+\sum\limits_{i_{2}=0}^{r-2}\alpha_{i_{2}}^{(1)}a_{i_{2}}a_{j_{1}}+\cdots+\sum\limits_{i_{k}=0}^{r-2}\alpha_{i_{k}}^{(k-1)}a_{i_{k}}a_{j_{1}}=0,\\
    &\text{ for all } 0\leq j_{1}\leq r-1;\\
	&\sum\limits_{i_{1}=0}^{r-1}\alpha_{i_{1}}^{(0)}a_{i_{1}}a_{j_{2}}+r\alpha_{j_{2}}^{(1)}(a_{j_{2}})^{2}+\cdots+\sum\limits_{i_{k}=0}^{r-2}\alpha_{i_{k}}^{(k-1)}a_{i_{k}}a_{j_{2}}=0,\\
    &\text{ for all } 0\leq j_{2}\leq r-2;\\
	&\hspace{44.5mm}\vdotswithin{=}\\	
	&\sum\limits_{i_{1}=0}^{r-1}\alpha_{i_{1}}^{(0)}a_{i_{1}}a_{j_{k}}+\sum\limits_{i_{2}=0}^{r-2}\alpha_{i_{2}}^{(1)}a_{i_{2}}a_{j_{k}}+\cdots+r\alpha_{j_{k}}^{(k-1)}(a_{j_{k}})^{2}=0,\\
    &\text{ for all } 0\leq j_{k}\leq r-2.
	\end{align*}	
    Let  $X^{(i)}=[\alpha^{(i)}_{0},\alpha^{(i)}_{1},\ldots,\alpha^{(i)}_{(r-2)}]^{T}$, for all {\small$1\leq i\leq k-1$} and $X^{(0)}=[\alpha^{(0)}_{0},\alpha^{(0)}_{1},\ldots,\alpha^{(0)}_{(r-1)}]^{T}$, and let
    \begin{align*}
        A&=[a_{ij}], \text{where } a_{ij}=ra_{i}^{2}\delta_{ij}, 0\leq i,j\leq r-1;\\
        B&=[b_{ij}], \text{where } b_{ij}=ra_{i}^{2}\delta_{ij}, 0\leq i,j\leq r-2;\\
        C&=[c_{ij}], \text{where }  c_{ij}=a_{i}a_{j},  0\leq i\leq r-1, 0\leq j\leq r-2;\\
        D&=[d_{ij}], \text{where }  d_{ij}=a_{i}a_{j}, 0\leq i\leq r-2,0\leq j\leq r-1;\\
        E&=[e_{ij}], \text{where } e_{ij}=a_{i}a_{j}, 0\leq i,j\leq r-2.
    \end{align*}
   The above system of linear equations can be written as follows:
   {\small \begin{equation}
       \begin{bmatrix}
           A&C&C&\cdots&C\\
           D&B&E&\cdots&E\\
           D&E&B&\cdots&E\\
           \vdots&\vdots&\vdots&\ddots&\vdots\\
           D&E&E&\cdots&B
       \end{bmatrix}\begin{bmatrix}
           X^{(0)}\\
           X^{(1)}\\
           X^{(2)}\\
           \vdots\\
            X^{(k-1)}
       \end{bmatrix}=\begin{bmatrix}
           0\\
           0\\
           0\\
           \vdots\\
            0
       \end{bmatrix}.
   \end{equation}}
	It is straightforward to check that 
   {\small \begin{equation}\label{deteq}
        \begin{vmatrix}
           A&C&C&\cdots&C\\
           D&B&E&\cdots&E\\
           D&E&B&\cdots&E\\
           \vdots&\vdots&\vdots&\ddots&\vdots\\
           D&E&E&\cdots&B 
        \end{vmatrix}= a_{r-1}\prod_{i=0}^{r-2}a_{i}^{k}\begin{vmatrix}
           A'&C'&C'&\cdots&C'\\
           D'&B'&E'&\cdots&E'\\
           D'&E'&B'&\cdots&E'\\
           \vdots&\vdots&\vdots&\ddots&\vdots\\
           D'&E'&E'&\cdots&B' 
        \end{vmatrix},
    \end{equation}}
    where
    \vspace{-0.25cm}
    \begin{align*}
         A'&=[a'_{ij}], \text{where } a'_{ij}=ra_{i}\delta_{ij}, 0\leq i,j\leq r-1;\\
          B'&=[b'_{ij}], \text{where } b'_{ij}=ra_{i}\delta_{ij},0\leq i,j\leq r-2;\\
           C'&=[c'_{ij}], \text{where } c'_{ij}=a_{j},0\leq i\leq r-1,0\leq j\leq r-2;\\
             D'&=[d'_{ij}], \text{where } d'_{ij}=a_{j}, 0\leq i\leq r-2,0\leq j\leq r-1;\\
                E'&=[e'_{ij}], \text{where } e'_{ij}=a_{j},  0\leq i,j\leq r-2.
    \end{align*}
By applying elementary row operations, we can verify that the determinant on the right-hand side of \eqref{deteq} is nonzero. This implies that the system of linear equations has only the zero vector as a solution, completing the proof.\end{proof}
\begin{corollary}\label{cor1}
	Let $B$ be any submatrix of $\mathbf{A}_{M}$ 
with distinct $\lambda$ $(\lambda\leq r)$ block-rows. Then, $\mathrm{gfrank}_{q}(B)=r+(\lambda-1)(r-1)$.
\end{corollary}
\begin{proof} Let $B$ be any submatrix of $\mathbf{A}_{M}$ consisting of $\lambda$ block-rows, and let
$K = [k_{ij}]_{0 \leq i \leq \lambda - 1,\ 0 \leq j \leq r - 1}$
be the corresponding model matrix of $B$. By applying elementary row and column operations to the model matrix $K$, we can obtain a new model matrix that corresponds to a submatrix of $\mathbf{H}$ comprising $\lambda$ distinct block-rows. This completes the proof.\end{proof}
\begin{remark}\label{remark.2}
   As observed earlier, the matrix $\mathbf{A}_{M}$ defined in \eqref{A} is rank-deficient. However, this deficiency can be avoided if the nonzero entries of $\mathbf{A}_{M}$ are chosen arbitrarily. In other words, when each $D^{(i)}$ is allowed to take distinct nonzero values for different indices $i$, the matrix may attain full rank, as illustrated by the following example over $\mathbb{F}_{11}$:
\begin{equation*}
H =
\begin{bmatrix}
4 & 0 & 0 & 1 & 0 & 0 & 2 & 0 & 0 \\
0 & 5 & 0 & 0 & 5 & 0 & 0 & 4 & 0 \\
0 & 0 & 6 & 0 & 0 & 1 & 0 & 0 & 5 \\
4 & 0 & 0 & 0 & 0 & 7 & 0 & 2 & 0 \\
0 & 3 & 0 & 5 & 0 & 0 & 0 & 0 & 7 \\
0 & 0 & 3 & 0 & 4 & 0 & 5 & 0 & 0 \\
4 & 0 & 0 & 0 & 6 & 0 & 0 & 0 & 1 \\
0 & 5 & 0 & 0 & 0 & 7 & 2 & 0 & 0 \\
0 & 0 & 2 & 6 & 0 & 0 & 0 & 1 & 0
\end{bmatrix}.
\end{equation*}
The significance of this special case lies in its ability to facilitate an easy determination of the exact rank of $\mathbf{A}_{M}$ and number of maximally entangled bits required for constructing entanglement-assisted quantum codes. In contrast, when the non-binary entries of $\mathbf{A}_{M}$ are populated randomly, the resulting matrix may become full rank. In such cases, determining the exact rank and the corresponding minimum number of required ebits analytically becomes non-trivial.
\end{remark}
\section{Non-binary EA-QC-QLDPC Codes}\label{Sec.3}
In this section, using \cite[Theorem 4]{galindo2019entanglement}, we construct two families of non-binary EA-QC-QLDPC codes. One of these families requires a single \textit{ebit}, and unassisted portion of its overall Tanner graph has a girth $>4$

Let $H_1$ and $H_2$ be submatrices of $\mathbf{A}_M$ in \eqref{A}, consisting of $\ell_1\geq1$ and $\ell_2\geq1$ block-rows from $\mathbf{A}_M$, respectively, s.t., $H_1$ and $H_2$ have no block-rows in common, and all block-rows within each matrix are distinct.
\begin{theorem}\label{theorem3}
	Let $\mathcal{C}_{1}$ and $\mathcal{C}_{2}$ be two non-binary classical QC-LDPC codes, and let $H_{1}$ and $H_{2}$ be their parity-check matrices, respectively. There exists a non-binary EA-QC-QLDPC code with the parameters $[[r^{2}, r^{2}-2r-(r-1)(\ell_{1}+\ell_{2}-2)+1;1]]_q$. Further, unassisted portion of the overall  Tanner graph of non-binary EA-QC-QLDPC code has girth greater than 4.
\end{theorem}
\begin{proof} By Remark~\ref{re1}, the 
required entangled bits \cite{galindo2019entanglement} are 
\begin{equation}\label{ebitseq}
    c=\text{gfrank}_{q}(H_{1}H_{2}^{T})=
    \text{gfrank}_{q}([T_{ij}]_{1\leq i\leq \ell_{1}, 1\leq j\leq \ell_{2}}),
\end{equation}
where
\begin{equation*}
   T_{ij}=\Theta= \begin{bmatrix}
        a_{0}a_{0} & a_{0}a_{1} & \cdots & a_{0}a_{(r-1)}\\
        a_{1}a_{0} & a_{1}a_{1}& \cdots & a_{1}a_{(r-1)}\\
        \vdots & \vdots& \ddots & \vdots\\
        a_{(r-1)}a_{0} & a_{(r-1)}a_{1}& \cdots & a_{(r-1)}a_{(r-1)}
    \end{bmatrix},
\end{equation*}
for all $1\leq i\leq \ell_{1}$ and $1\leq j\leq \ell_{2}$.
To compute the $\mathrm{gfrank}_{q}$ of the block matrix $[T_{ij}]$ in \eqref{ebitseq}, we perform the following elementary row operations: for all $2 \leq i \leq \ell_1$ and $1\leq j\leq \ell_{2}$, replace $T_{ij}$ with $T_{ij} - T_{1j}$. Next, for each $j$, in the block $T_{1j}$, replace the $k^{\text{th}}$ row (for $1 \leq k \leq r-1$) by $R_k \leftarrow R_k - a_k a_0^{-1} R_0$, where $a_0^{-1}$ is the multiplicative inverse of $a_0$ in the finite field $\mathbb{F}_q$. Consequently, we have $\mathrm{gfrank}_{q}(H_1 H_2^T) = 1$. By Corollary~\ref{cor1}, the code dimension follows directly. By design the matrix $[H_{1}^{T}|H_{2}^{T}]^{T}$ is the submatrix of \( \mathbf{A}_M \) which is $4$ cycles free, completing the proof.
\end{proof}
\begin{example}
  For $r=7$, let $D^{(0)}=\mathrm{diag}\{a_{0},a_{1},\ldots,a_{6}\}$, where $a_{i}\in \mathbb{F}_{8},$ for all $i$. Let $H_{1}= [D^{\mathrm{mod}(ij, 7)}]_{0\leq i\leq 2,0\leq j\leq 6}$ and  $H_{2}=[D^{\mathrm{mod}(ij, 7)}]_{4\leq i\leq 6,0\leq j\leq 6}$
be two parity-check matrices of non-binary classical QC-LDPC codes $\mathcal{C}_{1}$ and $\mathcal{C}_{2}$, respectively. By Theorem~\ref{theorem3}, there exists a non-binary EA-QC-QLDPC code with parameters $[[49,12;1]]_8$.
\end{example}
Next, we turn our attention to the construction of non-binary entanglement-assisted QC-QLDPC codes derived from a single classical non-binary QC-LDPC code. Before presenting the main result, we establish a lemma that is essential for its proof.
\begin{lemma}\label{L}
Let $r$ be an odd prime, and let \begin{equation*}
   L= \begin{bmatrix}
        (r-1)a_{0} & -a_{1}& \cdots & -a_{r-1}\\
        -a_{0} & (r-1)a_{1}& \cdots & -a_{r-1}\\
        \vdots & \vdots&\vdots& \ddots & \vdots\\
         -a_{0} & -a_{1}& \cdots & (r-1)a_{r-1}
    \end{bmatrix},
\end{equation*}  where $a_{i}\in\mathbb{F}_{q}$ for all $i$. Then, $\mathrm{gfrank}_{q}(L) =(r-1)$.
\end{lemma}
\begin{proof} For $0\leq i\leq (r-1)$, let $R_{i}$ denote $i^{\text{th}}$ row of $L$. Now, for each $2\leq i\leq (r-1)$, replacing the $i^{\text{th}}$ row $R_{i}$,  by $R_{i}-R_{2}$ in the given matrix $L$ results in the following matrix $L_{1}$:
\begin{equation*}
L_{1} =
\begin{bmatrix}
(r-1)a_{0} & -a_{1} & -a_{2} & -a_{3} & \cdots & -a_{r-1} \\
-a_{0} & (r-1)a_{1} & -a_{2} & -a_{3} & \cdots & -a_{r-1} \\
0 & -ra_{1} & ra_{2} & 0 & \cdots & 0 \\
0 & -ra_{1} & 0 & ra_{3} & \cdots & 0 \\
\vdots & \vdots & \vdots & \vdots & \ddots & \vdots \\
0 & -ra_{1} & 0 & 0 & \cdots & ra_{r-1}
\end{bmatrix}.
\end{equation*}
For each $2\leq i\leq (r-1)$, replacing the $i^{\text{th}}$ row $R_{i}$ by $r^{-1}R_{i}$ in the given matrix $L_{1}$ results in the following matrix $L_{2}$:
\begin{equation*}
L_{2} =
\begin{bmatrix}
(r-1)a_{0} & -a_{1} & -a_{2} & -a_{3} & \cdots & -a_{r-1} \\
-a_{0} & (r-1)a_{1} & -a_{2} & -a_{3} & \cdots & -a_{r-1} \\
0 & -a_{1} & a_{2} & 0 & \cdots & 0 \\
0 & -a_{1} & 0 & a_{3} & \cdots & 0 \\
\vdots & \vdots & \vdots & \vdots & \ddots & \vdots \\
0 & -a_{1} & 0 & 0 & \cdots & a_{r-1}
\end{bmatrix}.
\end{equation*}
Next, by replacing $R_{0}$ and $R_{1}$ in the matrix $L_{2}$ with $R_{0} + \sum_{i=2}^{r-1} R_{i}$ and $R_{1} + \sum_{i=2}^{r-1} R_{i}$, respectively, we obtain
\begin{equation*}
L_{3} =
\begin{bmatrix}
(r-1)a_{0} & -(r-1)a_{1} & 0 & 0 & \cdots & 0 \\
-a_{0} & a_{1} & 0 & 0 & \cdots & 0 \\
0 & -a_{1} & a_{2} & 0 & \cdots & 0 \\
0 & -a_{1} & 0 & a_{3} & \cdots & 0 \\
\vdots & \vdots & \vdots & \vdots & \ddots & \vdots \\
0 & -a_{1} & 0 & 0 & \cdots & a_{r-1}
\end{bmatrix}.
\end{equation*}
For each $2\leq i\leq (r-2)$, replacing $i^{\text{th}}$  row $R_{i}$ in $L_{3}$  by $R_{i}-R_{r-1}$ results in the following matrix:
\begin{equation*}
L_{4} =
\begin{bmatrix}
(r-1)a_{0} & -(r-1)a_{1} & 0 & 0 & \cdots & 0 \\
-a_{0} & a_{1} & 0 & 0 & \cdots & 0 \\
0 & 0 & a_{2} & 0 & \cdots & -a_{r-1}\\
0 & 0 & 0 & a_{3} & \cdots & -a_{r-1} \\
\vdots & \vdots & \vdots & \vdots & \ddots & \vdots \\
0 & -a_{1}& 0 & 0 & \cdots & a_{r-1}
\end{bmatrix}.
\end{equation*}
Whether the characteristic $p$ of field $\mathbb{F}_{q}$ divides the circulant size $r$ or not, in both cases, it is straightforward to see that $\operatorname{gfrank}_{q}(L_{4}) = r - 1$; hence, proved.
\end{proof}
\begin{theorem}\label{theorem4}
Let $C$ be a classical non-binary QC-LDPC code with the parity check matrix $H$ that contains $\ell$ distinct number of block-rows from the matrix $\mathbf{A}_{M}$ defined in \eqref{A} such that $2\ell<r$. Then, there exists a non-binary EA-QC-QLDPC code with parameters $[[r^{2},(r-1)(r-\ell+1);r+(\ell-1)(r-1)]]_q$. Moreover, the Tanner graph of the code $C$ has girth $> 4$.    
\end{theorem}
\begin{proof}
By employing Corollary \ref{cor1}, we obtain $\mathrm{gfrank}_{q}(H)$. Let
\begin{equation*}
    H=\begin{bmatrix}        D^{k_{1}x_{0}}&D^{k_{1}x_{1}}&\cdots&D^{k_{1}x_{(r-1)}}\\
    D^{k_{2}x_{0}}&D^{k_{2}x_{1}}&\cdots&D^{k_{2}x_{(r-1)}}\\
    \vdots&\vdots&\ddots&\vdots\\
    D^{k_{\ell}x_{0}}&D^{k_{\ell}x_{1}}&\cdots&D^{k_{\ell}x_{(r-1)}}
    \end{bmatrix},
\end{equation*}
where $\mathbb{F}_{r}=\{x_{0},x_{1},\ldots,x_{r-1}\}$ and $k_{i}\in\mathbb{F}_{r}$ for all $1\leq i\leq \ell$ and $k_{i}\neq k_{j}$ whenever $i\neq j$. A direct application of Lemma~\ref{lemma1} yields the following
\begin{equation}
   HH^{T}= \begin{bmatrix}
    D_{r}&\Theta&\cdots&\Theta\\
     \Theta&D_{r}&\cdots&\Theta\\
      \vdots&\vdots&\ddots&\vdots\\
       \Theta&\Theta&\cdots&D_{r}
\end{bmatrix}_{(\ell r)\times (\ell r)},~ \text{where}
\end{equation}
 \begin{equation*}
   \Theta= \begin{bmatrix}
        a_{0}a_{0} & a_{0}a_{1} & \cdots & a_{0}a_{(r-1)}\\
        a_{1}a_{0} & a_{1}a_{1}& \cdots & a_{1}a_{(r-1)}\\
        \vdots & \vdots& \ddots & \vdots\\
        a_{(r-1)}a_{0} & a_{(r-1)}a_{1}& \cdots & a_{(r-1)}a_{(r-1)}
    \end{bmatrix}, 
\end{equation*}
\begin{equation*}
   D_{r}= \begin{bmatrix}
        ra_{0}^{2} & 0& \cdots & 0\\
        0 & ra_{1}^{2}& \cdots & 0\\
        \vdots & \vdots& \ddots & \vdots\\
        0 & 0& \cdots & ra_{r-1}^{2}
    \end{bmatrix}.
\end{equation*}
For $0\leq u\leq (r-1)$ and $0 \leq v\leq (\ell-1)$,  $R_{u}^{(v)}$ represents $u^{\mathrm{th}}$ row of $v^{\mathrm{th}}$  block-row of the matrix $HH^{T}$. For each $0\leq u\leq (r-1)$ and $1 \leq v\leq (\ell-1)$, we replace $R_{u}^{(v)}$  by $(a_{u}^{-1}R_{u}^{(v)}-r^{-1}\sum_{u=0}^{(r-1)}a_{u}^{-1}R_{u}^{(0)})$ in the matrix $HH^{T}$ to obtain 
\begin{equation}\label{hht}
  \begin{bmatrix} D_{r}&\Theta&\cdots&\Theta\\    \mathbf{0}&L&\cdots&\mathbf{0}\\  
     \vdots&\vdots&\ddots&\vdots\\
      \mathbf{0}&\mathbf{0}&\cdots&L
\end{bmatrix}_{(\ell r)\times (\ell r)}, 
\end{equation}
where $\mathbf{0}$ represents the $r\times r$ zero matrix and 
 {\small \begin{equation*}
   L= \begin{bmatrix}
        (r-1)a_{0} & -a_{1}& \cdots & -a_{r-1}\\
        -a_{0} & (r-1)a_{1}& \cdots & -a_{r-1}\\
        \vdots & \vdots& \ddots & \vdots\\
         -a_{0} & -a_{1}& \cdots & (r-1)a_{r-1}
    \end{bmatrix}. 
\end{equation*}}
By Lemma \ref{L}, one can conclude that $\mathrm{gfrank}_{q}(HH^{T})$ is $r+(l-1)(r-1)$, since in the block matrix in equation \eqref{hht} the matrix $D_{r}$ is of full rank and the matrix $L$ has gfrank $(r-1)$ over $\mathbb{F}_{q}$. By construction, the Tanner graph of the QC-LDPC code has girth $>4$. This completes the proof.
\end{proof}
\begin{remark} We note the following: 
\begin{enumerate}[leftmargin=*]
   \item The proposed non-binary entanglement-assisted QC-QLDPC codes are based on the diagonal matrix $D^{(0)}=\text{diag}\{a_{0},a_{1},\ldots,a_{n-1}\}, \text{ where } a_{i}\in\mathbb{F}_{q}\setminus\{0\}$ and its circulant shifts. The proposed construction provides a structured framework for generating codes with analysis of their algebraic properties, particularly towards computing the rank of the associated matrices towards exact code rate analysis. 
   \item We can further generalize the proposed ideas to a broader class of non-binary entanglement-assisted QC-LDPC codes, in which each circulant matrix is allowed to have a different set of distinct nonzero entries. Such a generalization offers greater design flexibility and may lead to codes with improved performance characteristics under suitable parameter choices. However, in this more general setting, the rank analysis of the resulting matrices becomes significantly more challenging due to the increased structural complexity and reduced symmetry. A detailed investigation of this case is beyond the scope of the present work and is omitted here.
\end{enumerate}
\end{remark}
\section{Conclusions}\label{Sec.5}
We constructed non-binary QC-LDPC codes by tiling permutation matrices of prime order and  derived the exact code rate. This construction was used to derive two families of non-binary EA-QC-QLDPC codes, one of the families requires only a \textit{single ebit}. We analyzed the rank of the parity-check matrix of proposed non-binary QC-LDPC code that yielded a rank-deficient matrix. The motivation for considering specific non-binary QC-LDPC codes lies in the fact that it allows us to easily determine the exact code rate and ebits required to construct entanglement-assisted quantum codes. On the other hand, populating non-binary entries in the non-binary QC parity-check matrix may result in a full-rank matrix, as given in Remark \ref{remark.2}. Obtaining the exact rank and optimal number of ebits for this case is non-trivial and is part of future work.
\section*{Acknowledgment}
P. Kumar is supported by a post-doctoral fellowship from Anusandhan National Research Foundation, Govt. of India through the grant SERB/F/3132/2023-2024 to S. S. Garani.
\bibliography{mybibliography}
\end{document}